\documentclass[12 pt]{amsart}

\usepackage{graphicx}
\usepackage{amsmath,amsfonts,mathrsfs,amssymb}
\usepackage{graphicx}

\newtheorem{thm}{Theorem}[section]

\newtheorem{lem}[thm]{Lemma}
\renewcommand\thethm {\thesection.\the\value{thm}}

\begin{document}

\title{Embedding of Hypercube into Cylinder
}




\author{Weixing Ji}
\author{Qinghui Liu}
\address{
              Dept. Computer Science,
              Beijing Institute of Technology,
              Beijing 100081, PR China.
              qhliu@bit.edu.cn,\ jwx@bit.edu.cn,\ gzwang@bit.edu.cn,\ chaosdefinition@hotmail.com.
           }

\author{Guizhen Wang}
\author{Zhuojia Shen}


\maketitle

\begin{abstract}
Task mapping in modern high performance parallel computers can be
modeled as a graph embedding problem, which simulates the mapping as
embedding one graph into another and try to find the minimum
wirelength for the mapping. Though embedding problems have been
considered for several regular graphs, such as hypercubes into
grids, binary trees into grids, et al, it is still an open problem
for hypercubes into cylinders. In this paper, we consider the
problem of embedding hypercubes into cylinders to minimize the
wirelength. We obtain the exact wirelength formula of embedding
hypercube $Q^r$ into cylinder $C_{2^3}\times P_{2^{r-3}}$ with
$r\ge3$.

\end{abstract}

\section{Introduction}
\label{intro} On one hand, a parallel program can be modeled as a
task graph, in which the vertices of the graph represent a computing
task, and the edges represent the communications among different
tasks. On the other hand, a massive parallel computer has a large
number of processing nodes that are connected together with an
interconnection network. One of the key problems of efficient
execution of parallel programs on these computers is how to find an
optimal mapping from computing tasks to processing nodes, so that
the communication overhead could be reduced when the tasks are run
in parallel. Without loss of generality, this problem can also be
modeled as an graph embedding problem, since task mapping is
actually try to find an optimal embedding of computing task graph
into the interconnection graph (network) with minimum link
congestion. Unfortunately, task mapping is proved as a NP hard
problem and heuristics algorithms are usually used to find an
approximate solution for a given application and interconnection.

Researchers have been working on graph embedding for years and
proposed a number of solution for regular graphs, such as hypercubes
into grids~\cite{MRRM}, binary trees into grids~\cite{OPAT}, honeycomb into hypercubes~\cite{DOINA}, grids into grids~\cite{ROTTGER}. This paper introduce a new technique to
estimate the wirelength of embedding hypercube $Q^r$ into cylinder $C_{2^3}\times P_{2^{r-3}}$ with $r\ge3$.

The rest of the paper is organized as follows. We show the existing
work on graph embedding in Section~\ref{sec:problem}.
Section~\ref{sec:gray} introduces the gray embedding and some useful
property, which is used afterwards in the wirelength calculation.
Section ~\ref{sec:com} discusses some composite sets with Cartesian
production structure. The wirelength calculation of hypercube into
cylinder is given in section ~\ref{sec:embedding}. Conclusion and
future work appear in Section ~\ref{sec:con}.


\section{Problem definition}
\label{sec:problem}

Let $G$ and $H$ be finite graphs with $n$ vertices. $V(G)$ and
$V(H)$ denote the vertex sets of $G$ and $H$ respectively. $E(G)$
and $E(H)$ denote the edge sets of G and H respectively. An
embedding [4] $f$ of $G$ into $H$ is defined as follows:

(i) $f$ is a bijective map from $V(G)$ to $V(H)$;

(ii) $f$ is a one-to-one map from $E(G)$ to $\{Path_f(u,v)\ :\
Path_f(u,v)$ is a path in $H$ between $f(u)$ and $f(v)\}$.

The edge congestion of an embedding $f$ of $G$ into $H$ is the
maximum number of edges of the graph $G$ that are embedded on any
single edge of $H$. Let $EC_f(e)$ denote the number of edges $(u,v)$
of $G$ such that $e$ is in the path $Path_f(u,v)$ between $f(u)$ and
$f(v)$ in $H$, in other words,
$$EC_f(e)=|\{(u,v)\in E(G)\ :\ e\in P_f(u,v)\}|.$$
For any $S\subset E(H)$, define
$$EC_f(S)=\sum_{e\in S} EC_f(e).$$
The edge congestion of an embedding $f$ of $G$ into $H$ is given by
$$EC_f(G,H)=\max_{e\in E(H)}EC_f(e).$$
The minimum edge congestion of G into H
$$EC(G,H)=\min_{f:G\rightarrow H}EC_f(G,H),$$
where the minimum is taken over all embeddings $f$ of $G$ into $H$.

The edge congestion problem of $G$ into $H$ is to find an embedding
of $G$ into $H$ that induces minimum edge congestion $EC(G,H)$.

The wirelength of an embedding $f$ of $G$ into $H$ is given by
$$WL_f(G,H)=\sum_{(u,v)\in E(G)}d_H(f(u),f(v))
=\sum_{e\in E(H)}EC_f(e).$$ where $d_H(f(u),f(v))$ denote the length
of the path $Path_f(u,v)$ in $H$. The wirelength of $G$ into $H$ is
defined as
$$WL(G,H)=\min_{f: G\rightarrow H} WL_f(G,H),$$
where the minimum is taken over all embeddings $f$ of $G$ into $H$.

The wirelength problem of $G$ into $H$ is to find an embedding of
$G$ into $H$ that induces minimum wirelength $WL(G,H)$.

Manuel et al(\cite{MARR}) find that the maximal subgraph problem
play an important role in solving wirelength problem. For a graph
$G$ and an integer $m$,
$$I_G(m)=\max_{A\subset V(G),\ |A|=m}|I_G(A)|,$$
where
$$I_G(A)=\{(u,v)\in E(G)\ :\ u,v\in V(G)\}.$$
A subset $A\subset V(G)$ is called optimal if $|I_G(A)|=I_G(|A|)|$.

The following lemmas are proved in \cite{MRRM}. Note that a set of
edges of $H$ is said to be an edge cut of $H$, if the removal of
these edges results in a disconnection of $H$.

\begin{lem}[Congestion Lemma]
Let $G$ be an $r$-regular graph and $f$ be an embedding of $G$ into
$H$. Let $S$ be an edge cut of $H$ such that the removal of edges of
$S$ leaves $H$ into $2$ components $H_1$,$H_2$ and let
$G_1=f^{-1}(H_1)$, $G_2=f^{-1}(H_2)$. Also $S$ satisfy the following
condition,

(i) For every edge $(a, b)\in G_i$, $i=1,2$, $Path_f(a,b)$ has no
edges in $S$.

(ii) For every edge $(a,b)$ in $G$ with $a\in G_1$ and $b\in G_2$,
$Path_f(a, b)$ has exactly one edge in $S$.

(iii) $G_1$ is optimal.

Then $EC_f(S)$ is minimum and $EC_f(S)=r|V (G_1)|-2 |E(G_1)|.$
\end{lem}

\begin{lem}[Partition Lemma]
Let $f: G\rightarrow H$ be an embedding. Let $\{S_1,
S_2,\cdots,S_p\}$ be a partition of $E(H)$ such that each $S_i$ is
an edge cut of $H$. Then
$$WL_f(G,H)=\sum_{i=1}^p WL_f(S_i).$$
\end{lem}

We will discuss the embedding of following graphs.

$Q^r$, the graph of the $r$-dimensional hypercube, has vertex-set
$\{0,1\}^r$, the $r$-fold Cartesian product of $\{0,1\}$. Thus
$|V(Q^r)|=2^r$. $Q^r$ has an edge between two vertices ($r$-tuple of
$0$s and $1$s) if they differ in exactly one entry.

The $1$-dimensional grid with $d\ge2$ vertices is denoted as $P_d$.
The $1$-dimensional cycle with $d\ge2$ vertices is denoted as $C_d$.
The $2$-dimensional grid is defined as $P_{d_1}\times P_{d_2}$ ,
where $d_i\ge2$ is an integer for each $i=1,2$. The cylinder
$C_{d_1}\times P_{d_2}$, where $d_1\ge2$ and $d_2\ge 1$, is a $P_{d_1}\times
P_{d_2}$ grid with a wraparound edge in each column (see e.g.,
Figure \ref{f2}). The torus $C_{d_1}\times C_{d_2}$, where
$d_1,d_2\ge 2$, is a $P_{d_1}\times P_{d_2}$ grid with a wraparound
edge in each column and a wrapround edge in each row.

It is conjectured that the wirelength of embedding hypercube $Q^r$
into cycle $C_{2^r}$ is $3\cdot 2^{2r-3}-2^{r-1}$. It is called CT
conjecture~\cite{Guu,CT,ECT,MRRM}. It is also conjectured in
\cite{MARR} such that the wirelength of embedding hypercube $Q^r$
into cylinder $C_{2^{r_1}}\times P_{2^{r_2}}$ with positive integers
$r_1+r_2=r$ is
$$2^{r_1}(2^{2r_2-1}-2^{r_2-1})+2^{r_2}(3\cdot 2^{2r_1-3}-2^{r_1-1}),$$
and the wirelength of embedding hypercube $Q^r$ into torus
$C_{2^{r_1}}\times C_{2^{r_2}}$ with positive integers $r_1+r_2=r$
is
$$2^{r_1}(3\times2^{2r_2-3}-2^{r_2-1})+2^{r_2}(3\cdot 2^{2r_1-3}-2^{r_1-1}).$$

Manuel et al (\cite{MARR}) verified the case of embedding $Q^r$ into
cylinder $C_{2^2}\times P_{2^{r-2}}$ for $r\ge2$. We prove in this
paper that

\begin{thm}\label{mainthm}
For any $r\ge3$, the wirelength of embedding $Q^r$ into cylinder
$C_{2^{3}}\times P_{2^{r-3}}$ is
$$2^{r_1}(2^{2r_2-1}-2^{r_2-1})+2^{r_2}(3\cdot 2^{2r_1-3}-2^{r_1-1}),$$
where $r_1=3$ and $r_2=r-3$.
\end{thm}

\noindent {\bf Remark 1.} Our argument for Theorem \ref{mainthm}
also valid for embedding of hypercube $Q^6$ into torus $C_8\times
C_8$. So
$$WL(Q^6,C_8\times C_8)=2\times 8\times 20=320.$$

\section{Gray Embedding}
\label{sec:gray} Grid embedding plays an important role in computer
architecture, and researchers believe that gray embedding minimize
wirelength of emdedding hypercube into cycles, cylinders and
torus~\cite{CT,MARR}.

To construct gray embedding, we give first the bijection of $V(Q^r)$
to $V(P_{2^{r_1}}\times P_{2^{r_2}})$.

Given $d>0$ and under Gray code list of $d$ bits, every code
corresponding to a number, e.g.,
$$g_d(0^d)=0,\ g_d(0^{d-1}1)=1,\ g_d(0^{d-2}11)=2,\cdots,\
g_d(10^{r-1})=2^d-1.$$

Define an embedding $gray$ from $Q^r$ into $C_{2^{r_1}}\times
P_{2^{r_2}}$ with $r_1,r_2\ge2$ and $r_1+r_2=r$. The vertices of
$C_{2^{r_1}}\times P_{2^{r_2}}$ have coordinates of the form
$(i,j)$, for $i=0,1,\cdots, 2^{r_1}-1$, $j=0,1,\cdots,2^{r_2}-1$.
Every vertex of $Q^r$ correspond to a string in $\{0,1\}^r$. Take
any $w\in\{0,1\}^r$, write $w=uv$, where $u\in\{0,1\}^{r_1}$,
$v\in\{0,1\}^{r_2}$, define
$$gray(w)=(g_{r_1}(u),g_{r_2}(v)),$$
which corresponding to a unique vertex in $V(P_{2^{r_1}}\times
P_{2^{r_2}})$.

As an example, we illustrate the embedding from $V(Q^4)$ to
$V(P_4\times P_4)$. By the above construction
$$g_2(00)=0,\ g_2(01)=1,\ g_2(11)=2,\ g_2(10)=3.$$
So we can directly get the map $gray$ as shown in Figure \ref{f1}.

\begin{figure}[ht]
\includegraphics[width=0.9\textwidth]{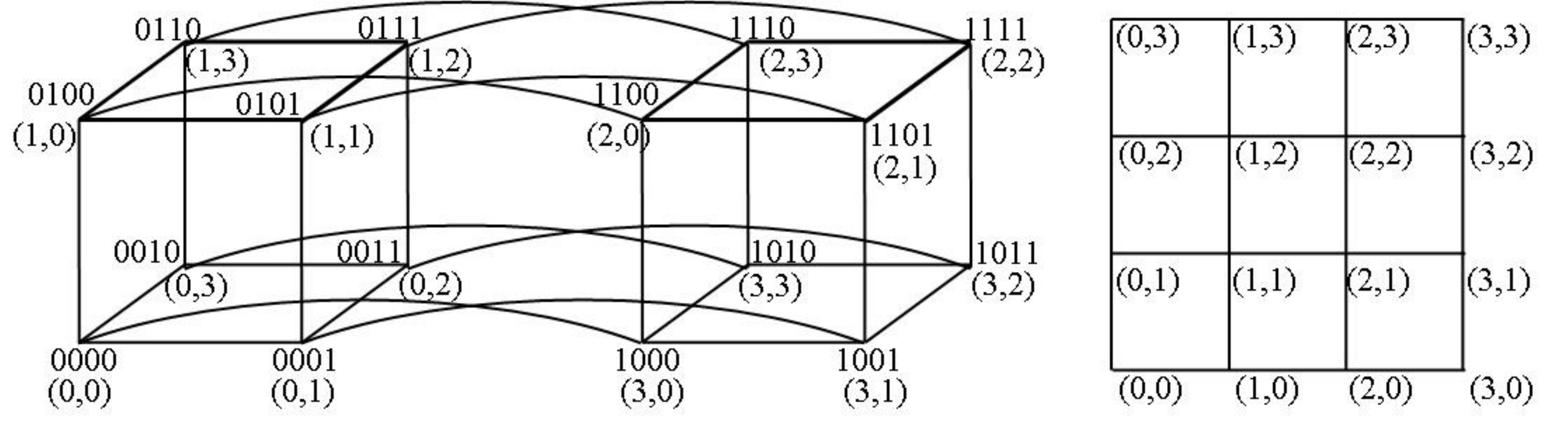}
\caption{\label{f1} The bijective map $gray:\ V(Q^4)\rightarrow
V(P_4\times P_4)$ }
\end{figure}

\begin{lem}\label{straight}
Fix $r\ge2$ and $r_1,r_2\ge1$ with $r_1+r_2=r$. Define a bijection
$gray$ from $V(Q^r)$ to $V(P_{2^{r_1}}\times P_{2^{r_2}})$ as above.
For any edge $(x,y)$ in hypercube $Q^r$, $gray(x)$ and $gray(y)$ are
in the same row or the same column of $P_{2^{r_1}}\times
P_{2^{r_2}}$.
\end{lem}

\begin{proof}
Without loose of generality, let $x,y\in\{0,1\}^r$. Write
$x=u_1v_1$, $y=u_2v_2$, with $|u_1|=|u_2|=r_1$ and
$|v_1|=|v_2|=r_2$. Since $x$ and $y$ are different only in one bit,
either $u_1=u_2$ or $v_1=v_2$. Hence either the x-coordinate or the
y-coordinate of $gray(x)$ and $gray(y)$ are equal. The result of the
lemma follows.\qed
\end{proof}

Let $G=Q^r$ and $H=C_{2^{r_1}}\times P_{2^{r_2}}$(respectively
$H=C_{2^{r_1}}\times C_{2^{r_2}}$). Now we can construct graph
embedding from $G$ to $H$. We have already define a vertices
bijection $gray$ from $G$ to $H$. For any edge $(x,y)$ in $G$,
define $Path_{gray}(x,y)$ be the path in $H$ between
$gray(x)$ and $gray(y)$ with minimal number of edges.

\section{Structure of a class of composite sets}
\label{sec:com}
In our paper we will discuss some composite sets with Cartesian
product structure.

A $c$-subcube of the $r$-cube is the subgraph of $Q^r$ induced by
the set of all vertices having the same value in some $r-c$
coordinates.

For any $0<k<2^r$, write
$$k=\sum_{i=1}^m 2^{c_i},\quad 0\le c_1<c_2<\cdots< c_m.$$
If $S$ is a subgraph of $Q^r$ which is a disjoint union of
$c_i$-subcubes, $1\le i\le m$, such that each $c_j$-subcube lies in
a neighborhood of every $c_i$-subcube for any $j>i$, then $S$ is
called a composite set (or cubal, see \cite{BGS,H}).

Let $S$ be a subgraph of $Q^r$ with $|V(S)|=k>0$. It is proved in
\cite{H} that $S$ is a composite set if and only if it is optimal,
or equivalently,
$$I_{Q^r}(S)=I_{Q^r}(k).$$

For any graph $S$ and $T$, recall that the Cartesian product
$S\times T$ of $S,T$ is, $V(S\times T)=V(S)\times V(T)$, and
$((u_1,u_2),(v_1,v_2))\in E(S\times T)$ if and only if
$$(u_1,v_1)\in E(S),\ u_2=v_2\in V(T)\quad \mbox{or}\quad
u_1=v_1\in V(S),\ (u_2,v_2)\in E(T).$$

Take any $r=r_1+r_2$ with $r_1,r_2\ge1$. By definition of Cartesian
product of graphs, it is direct to know that $Q^r=Q^{r_1}\times
Q^{r_2}$. If $S$ is a $d_1$-subcube of $Q^{r_1}$, $T$ is a
$d_2$-subcube of $Q^{r_2}$, then $S\times T$ is a
$(d_1+d_2)$-subcube of $Q^{r}$. Hence by definition of composite
set, we see that if $S$ is a subcube of $Q^{r_1}$, $T$ is a
composite set of $Q^{r_2}$, then $S\times T$ is a composite set of
$Q^{r}$. So we get the following lemma

\begin{lem}\label{cart}
Let $S$ and $T$ are composite sets of $Q^{r_1}$ and $Q^{r_2}$ respectively.
Suppose further that at least one of $S$ and $T$ is a subcube,
then $S\times T$ is a composite set of $Q^{r_1}\times Q^{r_2}$.
\end{lem}

Recall that the (binary-reflected) Gray code list for $d$ bits can
be generated recursively from the list for $d-1$ bits by reflecting
the list (i.e. listing the entries in reverse order), concatenating
the original list with the reversed list, prefixing the entries in
the original list with a binary $0$, and then prefixing the entries
in the reflected list with a binary $1$.

By this kind of recursive structure, we see that if $0\le k<
2^d$, then $g_d^{-1}(0:k)$ is a composite set, where we use the
matlab notation $m:n$ for any $0\le m\le n$ to indicate the set
$\{m,m+1,\cdots,n\}$ of consecutive integers.

In fact, write
$$k=\sum_{i=1}^m 2^{c_i},\quad 0\le c_1<c_2<\cdots< c_m.$$
Let $a_{m+1}=0$, and for $1\le i\le m$,
$$a_i=\sum_{k=i}^m 2^{c_k}.$$
Then for any $1\le i\le m$, $g_d^{-1}(a_{i+1}:a_{i}-1)$ is a
$c_i$-subcube in $Q^d$. These subcubes are disjoint, and each
$c_j$-subcube lies in a neighborhood of every $c_i$-subcube for any
$j>i$. So we get the following lemma

\begin{lem}\label{ch}
For any $d>0$ and $0\le j< 2^d$, $g_d^{-1}(0:j)$ is a composite set
of hypercube $Q^d$.
\end{lem}

\section{Hypercube into Cylinder}
\label{sec:embedding}
Now we can prove Theorem \ref{mainthm} by computing $WL(G,H)$ for
$G=Q^r$, $H=C_{8}\times P_{2^{r-3}}$. Denote $r_1=3$ and $r_2=r-3$.

To apply Congestion Lemma, we need to construct suitable edge cuts to form a partition.
For $j=1,2,\cdots,2^{r_2}-1$, define edge cut
$$B_j=\{((i,j-1),(i,j))\ :\ i=0,1,\cdots,7\}.$$
Define edge cuts
$$\begin{array}{rcl}
A_1&=&\{((0,j), (1,j)), ((3,j),(4,j))\ :\ j=0,1,\cdots, 2^{r_2}-1\}\\
A_2&=&\{((1,j), (2,j)), ((6,j),(7,j))\ :\ j=0,1,\cdots, 2^{r_2}-1\}\\
A_3&=&\{((2,j), (3,j)), ((5,j),(6,j))\ :\ j=0,1,\cdots, 2^{r_2}-1\}\\
A_4&=&\{((4,j), (5,j)), ((7,j),(0,j))\ :\ j=0,1,\cdots, 2^{r_2}-1\}\\
\end{array}
$$

For any $1\le i\le4$, $A_i$ disconnects $H$ into two components $X_i$ and
$X'_i$. For any $0\le j<2^{r_2}-1$, $B_j$ disconnects $H$ into two components $Y_j$
and $Y'_j$. We illustrate the case $r=6$, $r_1=r_2=3$ in Figure
\ref{f2}.

\begin{figure}[ht]
\includegraphics[width=0.9\textwidth]{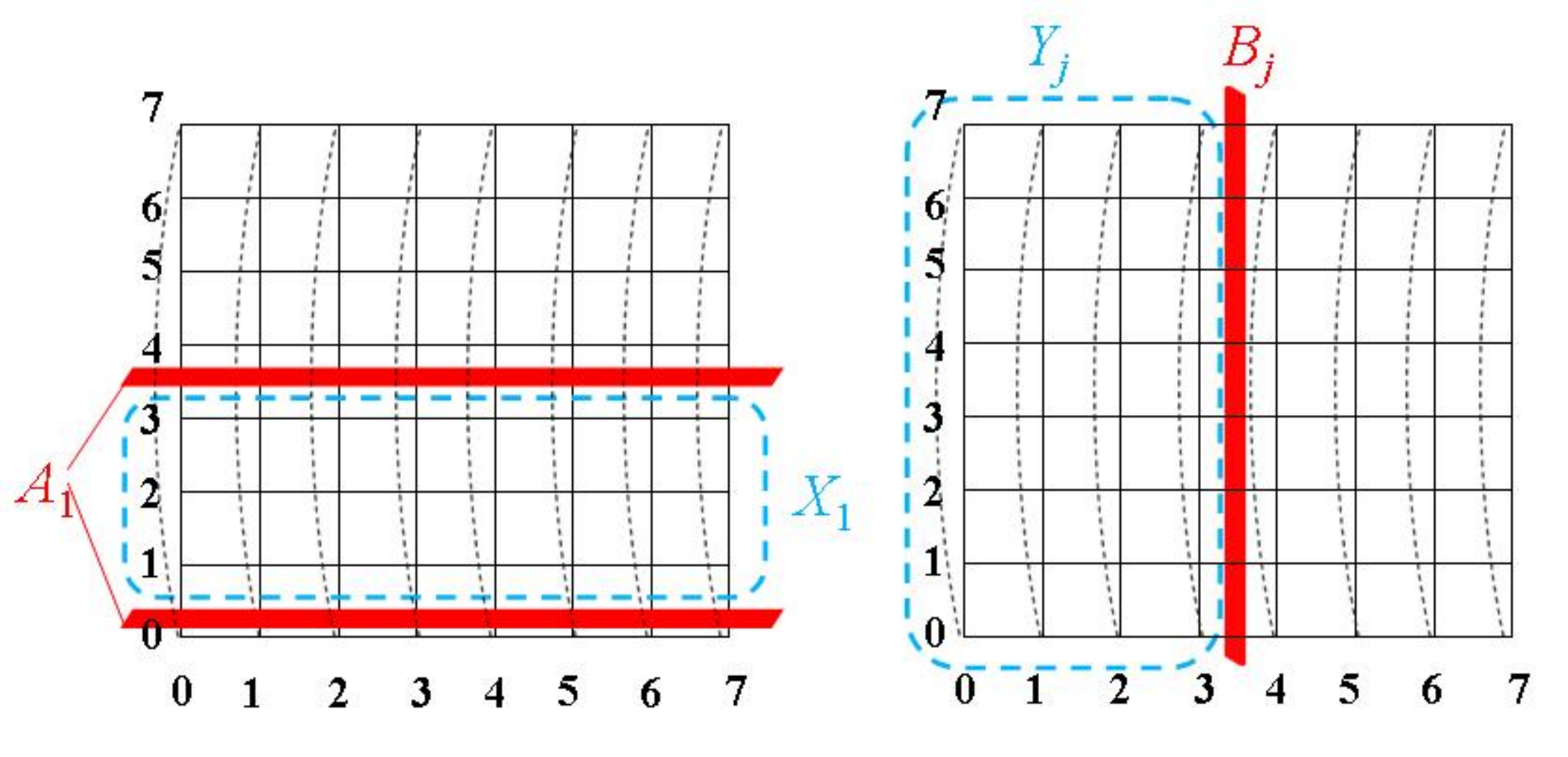}
\caption{\label{f2} The edge cut $A_1$ and $B_j$ in
$H=C_{2^{3}}\times P_{2^{r-3}}$}
\end{figure}

We discuss first $A_1$.

Let $G_1$ and $G'_1$ be the inverse images of $X_1$ and $X'_1$ under
$gray$ respectively. By Lemma \ref{straight}, the edge cut $A_1$
satisfies condition (ii) of the Congestion Lemma.

Now we show that $A_1$ satisfies condition (i). Since for any edge
$(u,v)\in G$, $Path_{gray}(u,v)$ is the shortest path connecting
$gray(u)$ and $gray(v)$, we only need to check whether there is a
path from $(0,j)$ to $(4,j)$ for any $j=0,1,\cdots,2^{r-3}-1$.
Notice that, if the Hamming distance of two codes is odd, then the
gray coding distance is odd, and vice versa. This implies that the
Hamming distance between $g_3^{-1}(0)$ and $g_3^{-1}(4)$ is even. So
there are no edge path connecting $(0,j)$ and $(4,j)$.

Notice that $V(X_{1})=\{1,2,3\}\times\{0,1,\cdots,2^{r_2}-1\}$ for
some $k$.
It is direct to see that $g_3^{-1}(\{1,2,3\})$ is a composite set,
and $g_{r_2}^{-1}\{0,1,\cdots,2^{r_2}-1\}$ is a $r_2$-subcube.
This implies that the subgraph $G_1$ is Cartesian product of a composite
set  and a subcube. By Lemma \ref{cart}, $G_1$
is a composite set, and hence optimal. Thus by the
Congestion Lemma, $EC_{gray}(A_1)$ is minimum.

The argument for $A_2$, $A_3$ and $A_4$ are analogous.
For $A_2$, we see that $g_3^{-1}(\{2,3,4,5,6\})$ is a composite set, and
there are no edge path connecting $(2,j)$ and $(6,j)$.
For $A_3$, we see that $g_3^{-1}(\{3,4,5\})$ is a composite set, and
there are no edge path connecting $(2,j)$ and $(6,j)$.
For $A_4$, we see that $g_3^{-1}(\{5,6,7\})$ is a composite set, and
there are no edge path connecting $(0,j)$ and $(4,j)$.
Thus by the Congestion Lemma, $EC_{gray}(A_2)$,
$EC_{gray}(A_3)$ and $EC_{gray}(A_4)$ are also minimum.

Fix any $1\le j\le 2^{r_2}-1$.
Let $G_j$ and $G'_j$ be the inverse images of $Y_j$ and $Y'_j$ under
$gray$ respectively. The edge cut $B_j$ satisfies conditions (i) and
(ii) of the Congestion Lemma. Note that
$V(Y_j)=\{0,1,\cdots,7\}\times \{0,1,\cdots,j-1\}$.
It is direct to see that $g_3^{-1}(0:7)$ is a $3$-subcube.
And by Lemma \ref{ch}, $g_{r_2}^{-1}(0:j)$ is a composite set.
This implies that $G_j$ is Cartesian product of a
sub-hypercube of order $3$ with a composite set.
By Lemma \ref{cart}, $G_j$
is also a composite set, and hence is optimal. Thus by the
Congestion Lemma, $EC_{gray}(B_j)$ is minimum.

The Partition lemma implies that $WL_{gray}(G,H)$ is minimum.

It is direct to compute that for any $r\ge3$
$$WL_{gray}(Q^r,C_{2^3}\times P_{2^{r-3}})=2^{r_1}(2^{2r_2-1}-2^{r_2-1})+2^{r_2}(3\cdot 2^{2r_1-3}-2^{r_1-1}),$$
where $r_1=3$ and $r_2=r-3$.

This proves Theorem \ref{mainthm}.

\section{Conclusion}
\label{sec:con}

Manuel et al get the exact wirelength of embedding of hypercube $Q^r$ into cylinder
$C_{2^2}\times P_{2^{r-3}}$ with $r\ge2$.
We prove in this paper that gray embedding minimizes the wirelength of embedding hypercube $Q^r$ into cylinder $C_{2^3}\times P_{2^{r-3}}$ with $r\ge3$, and hence get the exact wirelength of this case.
We also get the exact wirelength of embedding hypercube $Q^6$ into torus
$C_{2^3}\times C_{2^3}$.

We ever apply this method to study the case of embedding hypercube $Q^r$ into
cylinder $C_{2^4}\times P_{2^{r-4}}$ for any $r\ge4$.
We tried many embeddings(including gray embedding), but we can't get a partition to apply
Congestion Lemma.

\noindent {\bf Acknowledgements}
Liu and Wang are supported by the National Natural Science
Foundation of China, No. 11371055.  Ji is supported by the National
Natural Science Foundation of China, No. 61300010.



\end{document}